\pdfoutput=1
\RequirePackage[l2tabu,orthodox]{nag} % Warn about obsolete LaTeX usage

\documentclass[sigconf,authorversion]{acmart}

\copyrightyear{2017}
\acmYear{2017}
\setcopyright{acmlicensed}
\acmConference{ISSAC '17}{July 25-28, 2017}{Kaiserslautern, Germany}
\acmPrice{15.00}
\acmDOI{10.1145/3087604.3087642}
\acmISBN{978-1-4503-5064-8/17/07}
\usepackage[english]{babel}  % better hyphenation
\usepackage[shortlabels]{enumitem}
\usepackage{xspace}          % for \xspace = possible space after commands
\usepackage{color}
\usepackage{stmaryrd} % For llbracket and rrbracket

\usepackage{algorithm}
\usepackage[noend]{algpseudocode}

\algrenewcommand\alglinenumber[1]{{\scriptsize#1}}   %small line numbers
\algrenewcommand\algorithmicrequire{\textbf{Input:}} %instead of "Require"
\algrenewcommand\algorithmicensure{\textbf{Output:}} %instead of "Ensure"
\algnewcommand\algorithmicassume{\textbf{Requires:}} % designed for specifying assumption on input
\algrenewcommand\algorithmicfunction{\textbf{Fun}}
\algnewcommand\Assume{\item[\algorithmicassume]}
\newcommand{\ass}{\leftarrow}
\usepackage[capitalize]{cleveref}
\makeatletter
\def\renewtheorem#1{%
  \expandafter\let\csname#1\endcsname\relax
  \expandafter\let\csname c@#1\endcsname\relax
  \gdef\renewtheorem@envname{#1}
  \renewtheorem@secpar
}
\def\renewtheorem@secpar{\@ifnextchar[{\renewtheorem@numberedlike}{\renewtheorem@nonumberedlike}}
\def\renewtheorem@numberedlike[#1]#2{\newtheorem{\renewtheorem@envname}[#1]{#2}}
\def\renewtheorem@nonumberedlike#1{
\def\renewtheorem@caption{#1}
\edef\renewtheorem@nowithin{\noexpand\newtheorem{\renewtheorem@envname}{\renewtheorem@caption}}
\renewtheorem@thirdpar
}
\def\renewtheorem@thirdpar{\@ifnextchar[{\renewtheorem@within}{\renewtheorem@nowithin}}
\def\renewtheorem@within[#1]{\renewtheorem@nowithin[#1]}
\makeatother
\theoremstyle{acmplain}
\renewtheorem{theorem}{Theorem}[section]
\renewtheorem{proposition}[theorem]{Proposition}
\renewtheorem{lemma}[theorem]{Lemma}
\renewtheorem{definition}[theorem]{Definition}
\newtheorem{problem}{Problem} %% this one we added
\crefname{problem}{Problem}{Problems}
\Crefname{problem}{Problem}{Problems}
\newcommand{\ceil}[1]{\lceil #1 \rceil}

\newcommand\ZZ{\mathbb Z\xspace}
\newcommand\F{\mathbb F\xspace}

\newcommand\Osoft{O^{\scriptscriptstyle \sim}\!}
\newcommand{\K}{\mathbb{K}\xspace}
\newcommand{\Pox}{[\mkern-4.2mu[ x ]\mkern-4mu]}
\newcommand{\Kx}{\K\Pox}
\newcommand{\Pux}{\langle\mkern-4.2mu\langle x \rangle\mkern-4mu\rangle}
\newcommand{\R}{\Kx[y]}
\newcommand{\atzero}{_{|x=0}}
\newcommand{\valx}{v_x} %% {v_{(x=0)}}

\newcommand{\algoname}[1]{\ensuremath{\mathsf{#1}}\xspace}
\newcommand{\univRoots}{\mathsf{R_\K}\xspace}

\renewcommand{\le}{\leq}     % better-looking in my opinion (change if you like)
\renewcommand{\ge}{\geq}     % better-looking in my opinion (change if you like)
\DeclareMathOperator{\rem}{ rem }
\DeclareMathOperator{\quo}{ quo }
\newcommand{\val}{s}              % letter for valuation
\newcommand{\prc}{d}              % input precision
\newcommand{\pol}{Q}              % input polynomial
\newcommand{\pdeg}{n}             % degree of input polynomial
\newcommand{\polz}{\pol\atzero}   % input polynomial ``at zero''
\newcommand{\aff}{A}              % affine factor
\newcommand{\coaff}{B}            % affine cofactor
\newcommand{\coaffi}{C}            % affine cofactor
\newcommand{\rmd}{R}              % remainders polynomial
\newcommand{\rt}{f}               % root to some precision (power series)
\newcommand{\rtpol}{f}            % polynomial in x in root representation (\rtpol,\rtxpt)
\newcommand{\rtxpt}{t}            % shifting exponent in root representation (\rtpol,\rtxpt)
\newcommand{\nrts}{\ell}          % number of roots to some precision
\newcommand{\mult}{m}             % multiplicity
\newcommand{\alive}{\mathcal{I}}  % set for indices that are alive in affine factors algorithm
\newcommand{\ring}{\mathcal{A}}   % some abstract ring
\newcommand{\field}{\K}           % base field
\newcommand{\polRing}{\R}         % ring of univariate polynomials in y over \Kx = K[[x][y]
\newcommand{\ZZp}{\ZZ_{>0}}       % positive integers
\newcommand{\NN}{\ZZ_{\ge 0}}     % nonnegative integers
\newcommand{\series}{\mathcal{S}}         % power series in x
\newcommand{\trcSer}[1][\prc]{\series_{#1}}  % truncated series = K[[x]]/(x^d)
\newcommand{\trcSerPol}[1][\prc]{\trcSer[#1][y]}  % K[[x]]/(x^d)[y]
\newcommand{\storeArg}{} % aux, not to be used in document
\newcommand{\rtset}[1][\pol]{\renewcommand\storeArg{#1}\rtsetAux} % set of roots, 2 opt args (\rtpol,\rtxpt)
\newcommand{\rtsetAux}[1][\prc]{\mathcal{R}(\storeArg,#1)} % not to be used in document
\newcommand{\idealGen}[1]{(#1)}

\newcommand{\myparagraph}[1]{\paragraph{\hspace{-0.35cm}\textbf{#1}}}  % because acmart's paragraphs are invisible
\begin{document}
\title[Fast Computation of the Roots of Polynomials Over the Ring of Power Series]{Fast Computation of the Roots of Polynomials \\ Over the Ring of Power Series}
%\titlenote{Produces the permission block, and copyright information}
%\subtitle{Extended Abstract}
%\subtitlenote{The full version of the author's guide is available as
%  \texttt{acmart.pdf} document}

\author{Vincent Neiger}
%\authornote{ }
\affiliation{%
  \institution{Technical University of Denmark}
  %\streetaddress{P.O. Box 1212}
  \city{Kgs. Lyngby}
  \state{Denmark}
  %\postcode{}
}
\email{vinn@dtu.dk}

\author{Johan Rosenkilde}
%\authornote{ }
\affiliation{%
  \institution{Technical University of Denmark}
  %\streetaddress{P.O. Box 1212}
  \city{Kgs. Lyngby}
  \state{Denmark}
  %\postcode{}
}
\email{jsrn@jsrn.dk}

\author{\'Eric Schost}
%\authornote{ }
\affiliation{%
  \institution{University of Waterloo}
  %\streetaddress{P.O. Box 1212}
  \city{Waterloo ON}
  \state{Canada}
  %\postcode{43017-6221}
}
\email{eschost@uwaterloo.ca}

\begin{abstract}
  We give an algorithm for computing all roots of polynomials over a univariate power series ring over an exact field $\mathbb{K}$.
  More precisely, given a precision $d$, and a polynomial $Q$ whose coefficients are power series in $x$, the algorithm computes a representation of all power series $f(x)$ such that $Q(f(x)) = 0 \bmod x^d$.
  The algorithm works unconditionally, in particular also with multiple roots, where Newton iteration fails.
  Our main motivation comes from coding theory where instances of this problem arise and multiple roots must be handled.

  The cost bound for our algorithm matches the worst-case input and output size $d \deg(Q)$, up to logarithmic factors.
  This improves upon previous algorithms which were quadratic in at least one of $d$ and $\deg(Q)$.
  Our algorithm is a refinement of a divide \& conquer algorithm by Alekhnovich (2005), where the cost of recursive steps is better controlled via the computation of a factor of $Q$ which has a smaller degree while preserving the roots.
\end{abstract}


\maketitle

\section{Introduction}
\label{sec:intro}

In what follows, $\K$ is an exact field, and $\polRing$ denotes the set of
polynomials in $y$ whose coefficients are power series in $x$ over $\field$.

\myparagraph{Problem and main result}

Given a polynomial in $\polRing$, we are interested in computing its power
series roots to some precision, as defined below.

\begin{definition}
  \label{dfn:root}
  Let $\pol \in \R$ and $\prc \in \ZZ_{>0}$. A power series $\rt \in
  \Kx$ is called a \emph{root of $\pol$ to precision $\prc$} if
  $\pol(\rt) = 0 \bmod x^\prc$; the set of all such roots is denoted
  by $\rtset$.
\end{definition}

Our main problem (\cref{pbm:series_root}) asks, given $\pol$ and $\prc$, to compute a finite
representation of $\rtset$; the fact that such a representation exists is
explained below (\cref{thm:modular_roots_partition}). In all the paper, 
we count operations in $\K$ at unit cost, and we use the soft-O notation
$\Osoft(\cdot)$ to give asymptotic bounds with hidden polylogarithmic factors.

\begin{center}
\fbox{ \begin{minipage}{8cm}
\begin{problem}
  \label{pbm:series_root}
~\\
\emph{Input:}
\begin{itemize}
  \setlength{\itemsep}{0pt}
  \item a precision $\prc \in \ZZp$,
  \item a polynomial $\pol \in \polRing$ known at precision $\prc$.
\end{itemize}

\emph{Output:}
  \begin{itemize}
    \item (finite) list of pairs
      $(\rtpol_i, \rtxpt_i)_{1 \le i \le \nrts} \subset \K[x] \times \NN$
      such that
      $\rtset \,=\, \bigcup_{1\le i\le \nrts} (\rtpol_i + x^{\rtxpt_i} \Kx)$
  \end{itemize}
\end{problem}
\end{minipage}
}
\end{center}

\medskip

An algorithm solving this problem must involve finding roots of
polynomials in $\K[y]$.
The existence, and complexity, of root-finding algorithms for
univariate polynomials over $\K$ depends on the nature of $\K$. In
this paper, we assume that $\K$ is such that we can find roots in $\K$
of a degree $n$ polynomial in $\K[y]$ in time $\univRoots(n)$, for some
function $\univRoots: \NN \to \mathbb{R}$; the underlying algorithm may
be deterministic or randomized. For instance, if $\K=\F_q$, we can
take $\univRoots(n) \in \Osoft(n)$ using either a Las Vegas algorithm (in
which case the runtime can be more precisely stated as
$\Osoft(n\log(q))$~\cite[Cor.\,14.16]{von_zur_gathen_modern_2013}),
or a deterministic one (with for instance a runtime
$\Osoft(nk^2\sqrt{p})$, where we write $q=p^k$, $p$
prime~\cite{Shoup91}). 

We now state our main result: we separate the cost of the root-finding part of
the algorithm, which may be randomized, and the rest of the algorithm which is
deterministic.

\begin{theorem}
  \label{thm:series_root}
  There is an algorithm which solves \cref{pbm:series_root} using
  $\Osoft(\prc\pdeg)$ deterministic operations in $\field$, together
  with an extra $O(\prc\univRoots(\pdeg))$ operations, where $\pdeg
  = \deg(\pol)$.
\end{theorem}

A cost in $\Osoft(\prc\pdeg)$ is essentially optimal for
\cref{pbm:series_root}. Indeed, if $Q=(y-f_1)\cdots (y-f_n)$, for some
power series $f_1,\dots,f_n$ such that $f_i-f_j$ is a unit for all $i \ne j$,
then the roots of $Q$ to precision $d$ are all the power series of the
form $f_i +x^d \Kx$, for some $i$. In this case, solving \cref{pbm:series_root}
involves computing all $f_i \bmod x^d$, which amounts to
$\prc\pdeg$ elements in $\K$.

\myparagraph{Previous work} 

When the discriminant of $Q \in \Kx[y]$ has $x$-valuation zero, or
equivalently, when all $y$-roots of $Q\atzero$ are simple (as in the
example above), our problem admits an obvious solution: first, compute
all $y$-roots of $Q\atzero$ in $\K$, say $y_1,\dots,y_\ell$, for some
$\ell \le n$, where $n = \deg Q$.  Then, apply Newton iteration to
each of these roots to lift them to power series roots
$f_1,\dots,f_\ell$ of precision $d$; to go from precision say $d/2$ to
$d$, Newton iteration replaces $f_i$ by
\[
  f_i - \frac{Q(f_i)}{Q'(f_i)} \bmod x^{d} \ ,
\]
where $Q' \in \Kx[y]$ is the formal derivative of $Q$.
The bottleneck of this approach is the evaluation of all $Q(f_i)$ and $Q'(f_i)$.
Using an algorithm
for fast multi-point evaluation in the ring of univariate polynomials
over $\Kx/\idealGen{x^{d}}$, these evaluations can both be done in $\Osoft(dn)$
operations in $\K$.  Taking all steps into account, we obtain the
roots $f_1,\dots,f_\ell$ modulo $x^d$ using $\Osoft(dn)$ operations in
$\K$; this is essentially optimal, as we pointed out above. In this case,
the total time for root-finding is $\univRoots(n)$. 

Thus, the non-trivial cases of Problem~\ref{pbm:series_root} arise
when $Q\atzero$ has multiple roots. In this case, leaving aside the
cost of root-finding, which is handled in a non-uniform way in
previous work, we are not aware of an algorithm with a cost similar to ours.
The best cost bounds known to us are $\Osoft(n^2d)$, obtained
in~\cite{alekhnovich_linear_2005} and with this cost estimate being showed
in~\cite{nielsen_sub-quadratic_2015}, and $\Osoft(nd^2)$,
obtained in~\cite{berthomieu_polynomial_2013}.

When $Q\atzero$ has multiple roots, a natural generalization of our
problem consists in computing Puiseux series solutions of $Q$. It is
then customary to consider a two-stage computation: first, compute
sufficiently many terms of the power series / Puiseux series solutions
in order to be able to {\em separate} the branches, then switch to
another algorithm to compute many terms efficiently.

Most algorithms for the first stage compute the so-called singular
parts of rational Puiseux expansions~\cite{Duval89} of the solutions.
They are inspired by what we will call the {\em Newton-Puiseux}
algorithm, that is, Newton's algorithmic proof that the field of
Puiseux series $\K\Pux$ is algebraically
closed when $\K$ is algebraically closed of characteristic
zero~\cite{Newton1736,Walker78}. In the case of Puiseux series roots,
one starts by reading off the leading exponent $\gamma$ of a possible
solution on the Newton polygon of the input equation $Q \in
\K\Pux [y]$.
The algorithm then considers $\hat{Q} = Q(x^\gamma y)/x^s \in \K\Pux[y]$, where $s$ is the valuation at $x$ of $Q(x^\gamma y)$.
If $y_1,\ldots,y_\ell$ are the $y$-roots of ${\hat{Q}}\atzero$, then these give the $x^\gamma$ terms of the Puiseux series roots of $Q$.
For each $i$ we then replace $Q$ with $Q(x^\gamma (y_i + y)) / x^{s'}$, where $s'$ is the valuation at $x$ of $Q(x^\gamma(y_i + y))$.
This allows us to compute the terms of the solutions one by one.
The best algorithms to date~\cite{PoRy11,PoRy15} use an expected number of $\Osoft(n^2 \nu + n^3 +n^2 \log(q))$ operations in $\K$, if $\K = \F_q$ and where $\nu$ is the valuation of the discriminant of $Q$.
These algorithms are randomized of the Las Vegas type, since they rely on Las Vegas root-finding in $\F_q[y]$.

%% \jsrn{All of the methods discussed here relies on Las Vegas root-finding when applied to $\F_q$!}

In the second stage, given the singular parts of the solutions, it
becomes for instance possible to apply Newton iteration, as
in~\cite{KuTr78}. If $Q$ is actually in $\K[x][y]$, one may
alternatively derive from it a linear recurrence with polynomial
coefficients satisfied by the coefficients of the solutions we are
looking for; this allows us to compute them at precision $d$ using
$O(dn)$ operations, that is, in time genuinely linear in
$n,d$~\cite{ChCh86a,Chch86b} (keeping in mind that in both cases, we
may need to know about $\nu$ terms of the solutions before being able
to switch to the faster algorithm). We will discuss a similar observation
in the context of our algorithm, in \cref{sec:alg:roots}.

%% \jsrn{Isn't this basically the fact that, once you know enough terms of a root $f$, then the shift $Q(f + x^t y)$ is going to have degree 1 at $x=0$. Then the affine factor has degree 1, i.e.~$A = A_0 + y A_1$, so the rest of $f$ is simply $-A_0/A_1$.
%%   I have a todo about making this remark at the end of the paper. Should we link these things here?
%% }

Using ideas akin to the Newton-Puiseux algorithm, Berthomieu, Lecerf,
and Quintin gave in~\cite{berthomieu_polynomial_2013} an algorithm that
computes roots of polynomials in $L[y]$, for a wide class of local
rings $L$. In the particular case $L=\F_q\llbracket
x\rrbracket$ with $q=p^s$, the expected runtime of their algorithm is
$\Osoft(n d^2 + n \log(q)+ nd \log(k)/p)$ operations in~$\F_q$.

Let us finally mention algorithms for polynomial factorization over
local fields. Using the Montes algorithm~\cite{Montes99},
it is proved in~\cite{BaNaSt13} that one can compute a so-called
OM-factorization of a degree $n$ polynomial $Q$ in $\F_q\Pux[y]$ at
precision $d$ using $\Osoft(n^2\nu+n \nu^2 + n\nu\log(q))$, where
$\nu$ is the valuation of the discriminant of $Q$; the relation to
\emph{basic root sets}, defined below, remains to be elucidated.

%% \todo{discuss: factor at precision d? relation with roots at precision d?}.

Sudan's and Guruswami-Sudan's algorithms for the list-decoding of
Reed-Solomon codes~\cite{sudan_decoding_1997,guruswami_improved_1999}
have inspired a large body of work, some of which is directly related
to Problem~\ref{pbm:series_root}. These algorithms operate in two stages: the
first stage finds a polynomial in $\K[x,y]$ with some constraints; the second one finds
its factors of the form $y-f(x)$, for $f$ in $\K[x]$.

The Newton-Puiseux algorithm can easily be adapted to compute such
factors; in this context, it becomes essentially what is known as the
Roth-Ruckenstein algorithm~\cite{roth_efficient_2000}; its cost
is in $O(d^2n^2)$, omitting the work for univariate root-finding.

In the context of Sudan's and Guruswami-Sudan's algorithms, we may
actually be able to use Newton iteration directly, by exploiting the
fact that we are looking for {\em polynomial} roots.  Instead of
computing power series solutions (that is, the Taylor expansions of
these polynomial roots at the origin), one can as well start from
another expansion point $x_0$ in $\K$; if the discriminant of $Q$ does
not vanish at $x_0$, Newton iteration applies. If $\K$ is finite, one
cannot exclude the possibility that all $x_0$ in $\K$ are roots of
$Q$; if needed, one may then look for $x_0$ in an extension of $\K$ of
small degree. Augot and Pecquet showed in~\cite{augot_hensel_2000} that in the cases appearing in Sudan's algorithm, there is always a suitable $x_0$ in $\K$.

However, for example for the Wu list decoding algorithm \cite{wu_new_2008}
or for the list-decoding of certain algebraic geometry codes
\cite{nielsen_sub-quadratic_2015}, one does seek truncated power
series roots. In this case, one may use Alekhnovich's
algorithm~\cite[App.]{alekhnovich_linear_2005}, which is a divide
and conquer variant of the Roth-Ruckenstein algorithm. It solves
Problem~\ref{pbm:series_root} using $n^{O(1)} \Osoft(d)$
operations in $\K$ plus calls to univariate root-finding; the refined analysis
in~\cite{nielsen_sub-quadratic_2015} gives the runtime
$\Osoft(n^2 d + nd\log q)$.

\myparagraph{Outline} We start by giving properties about the
structure of the set of roots in \cref{sec:structure_roots}. We will
see in particular how $\rtset$ can be described recursively as the
finite union of set of roots at a lower precision for shifts of
$\pol$, that is, polynomials of the form $\pol(\rtpol+x^{\rtxpt}y)$.
From this, we will be able to derive a divide-and-conquer algorithm
which is essentially Alekhnovich's. 

The reason why the runtime of this algorithm is quadratic in $n$ is
the growth of the (sum of the) degrees of these shifts.  Having in
mind to control this degree growth, we conclude
\cref{sec:structure_roots} with the definition of so-called
\emph{reduced root sets}, for which we establish useful degree properties.

In \cref{sec:affine_factors}, we detail a fast algorithm for the
computation of \emph{affine factors}, which are polynomials having the
same roots as the shifts but which can be computed more efficiently
thanks to the degree properties of our reduced root sets. Finally, in
\cref{sec:alg:roots}, we incorporate this into the divide and conquer
approach, leading to our fast power series roots algorithm.

\section{Structure of the set of roots}
\label{sec:structure_roots}

Recall the notation of \cref{pbm:series_root}.  In the following
analysis, we consider knowing $\pol$ to arbitrary precision,
i.e.~$\pol \in \R$. For convenience, we also define for any $d \leq
0$ that $\rtset = \Kx$. First, we introduce basic notation.

\begin{itemize}
  \item $\valx: \R\setminus\{0\} \rightarrow \NN$ denotes the valuation at $x$, that is,
    $\valx(\pol)$ is the greatest power of $x$ which divides $\pol$, for any nonzero
    $\pol \in \R$.
  \item For $\pol \in \R$, we write $\polz$ for the univariate polynomial in
    $\K[y]$ obtained by replacing $x$ by $0$ in $\pol$.
  \item We denote by $\trcSer = \Kx/\idealGen{x^\prc}$ the ring of power series in
$x$ over $\field$ truncated at precision $\prc$.
  \item To avoid confusion, $\deg(\cdot)$ stands for the degree of some
    polynomial in $y$ over $\K$, over $\Kx$, or over $\trcSer$, whereas the
    degree of polynomials in $\K[x]$ is denoted using $\deg_x(\cdot)$.
\end{itemize}

The next lemma follows from the above definitions, and shows that we can focus
on the case $\valx(\pol)=0$.

\begin{lemma}
  \label{lem:modular_roots_valuation}
  Let $\pol \in \R$ be nonzero and let $\prc \in \ZZ_{>0}$. If $\polz =
  0$, then $\rtset = \rtset[x^{-\val}\pol][\,\prc-\val]$, where $\val =
  \valx(\pol)$.
\end{lemma}
%\begin{proof}
%  Let $\hat Q = Q/x^s$ i.e.~$Q = x^s\hat Q$.
%  Then for any $f \in \Kx$:
%  \[
%    \hat Q(f) = 0 \mod x^{d-s} \iff Q(f) = x^s\hat Q(f) = 0 \mod x^d \ .
%  \]
%\end{proof}

Now, we will focus on a compact way of representing root sets, and we will see
that $\rtset$ always admit such a representation even though it is usually an
infinite set. Similar representations are also behind the correctness and the
efficiency of the algorithms of Roth-Ruckenstein \cite{roth_efficient_2000}, of
Alekhnovich \cite[App.]{alekhnovich_linear_2005}, and of
Berthomieu-Lecerf-Quintin \cite[Sec.\,2.2]{berthomieu_polynomial_2013}. To
support the divide-and-conquer structure of our algorithm, we further describe
how these representations compose.

\begin{definition}\label{def:basic_root_set}
  Let $\pol \in \R$ be nonzero and let $\prc \in \ZZ_{>0}$.  A \emph{basic root
  set} of $\pol$ to precision $d$ is a finite set of pairs $(f_i, t_i)_{1\le
  i\le \ell}$, each in $\K[x] \times \ZZ_{\geq 0}$, such that:
  \begin{itemize}
  \item $\valx(Q(f_i + x^{t_i} y)) \geq d$ for $1\le i\le\ell$,
  \item we have the identity
  \[
  \rtset = \bigcup_{1\le i \le \ell} \left \{ f_i + x^{t_i} \Kx \right \}.
  \]
  \end{itemize}
  For $d \le 0$, we define the unique basic root set of $\pol$ to precision $d$
  as being $\{(0,0)\}$; note that it satisfies both conditions above.
\end{definition}
We remark that the first restriction on being a basic root set is key:
for instance, $Q = y^2 + y \in \F_2\Pox[y]$ has $\rtset[Q][1] =
\F_2\Pox$. But $\{(0,0)\}$ is \emph{not} a basic root set because it
does not satisfy the first property; rather a basic root set is given
by expanding the first coefficient: $\{ (0,1), (1,1) \}$.

At precision $d=1$, one can easily build a basic root set of $Q$ which has
small cardinality:
\begin{lemma}
   \label{lem:root_set_1}
   Let $\pol \in \R$ be such that $\polz \neq 0$, and let
   $y_1,\ldots,y_\ell$ be the roots of $\pol\atzero$. Then,
   $(y_i,1)_{1 \le i \le \ell}$ is a basic root set of $Q$ to
   precision $1$.
\end{lemma}
\begin{proof}
  Take $i$ in $\{1,\dots,\ell\}$ and write the Taylor expansion
  of $Q(y_i+xy)$ as  $Q(y_i+xy)=Q(y_i) + xR_i(y)$, for some 
  $R_i\in\R$. Since both terms in the sum have valuation
  at least $1$, we obtain that $s_i=\valx(Q(y_i+xy))$ is at least $1$.
  Furthermore, we remark that 
  \begin{align*}
    \rtset[\pol][1] & = \{f \in \Kx \mid \pol(f)=0 \bmod x\}  \\
                    & = \{f \in \Kx \mid \polz(f_0)=0\},
  \end{align*}
  where $f_0$ is the constant coefficient of $f$. Thus, $\rtset[\pol][1]$ is
  the set of $f\in\Kx$ whose constant coefficient is in $\{y_1,\dots,y_\ell\}$.
\end{proof}

\begin{proposition}
  \label{prop:roots_dc}
  Let $\pol \in \R$ be such that $\polz \neq 0$ and let $ \prc', \prc$
  be in $\ZZ_{\ge 0}$, with $\prc'\le \prc$.  Suppose that $\pol$
  admits a basic root set $(f_i, t_i)_{1 \leq i \leq \ell}$ to
  precision $\prc'$. Suppose furthermore
  that, for $1 \le i \le \ell$, $\pol(f_i + x^{t_i} y)/x^{s_i}$ admits a basic
  root set $(f_{i,j}, t_{i,j})_{1 \leq j \leq \ell_i}$ to precision $\prc -
  s_i$, where $s_i = \valx(\pol(f_i + x^{t_i} y))$. Then, a basic root
  set of $\pol$ to precision $\prc$ is given by
  \[
    (f_i + f_{i,j}x^{t_i}, t_i + t_{i,j})_{1 \leq j \leq \ell_i, 1 \leq i \leq \ell} \ .
  \]
\end{proposition}
\begin{proof}
  For $1\le i \le \ell$, let  $Q_i = Q(f_i + x^{t_i} y)/x^{s_i}$.
  Then, for all $i,j$, from the definition of basic root sets, we have
  \begin{align*}
    \valx\Big(Q(f_i + f_{i,j}x^{t_i} + x^{t_i + t_{i,j}})\Big)
    &= \valx\Big(\big(x^{s_i} Q_i\big)_{|y = f_{i,j} + x^{t_{i,j}}y}\Big) \\
    &\geq s_i + (\prc - s_i).
  \end{align*}
  This proves that the first property of \cref{def:basic_root_set} holds.

  For the second property, we prove both inclusions leading to the identity
  $\rtset = \cup_{i,j} \{ f_i + x^{t_i}f_{i,j} + x^{t_i + t_{i,j}} \Kx \}$.

  First, consider some $f \in \rtset$; since $d' \le d$, $f$ is in $
  \rtset[\pol][d']$,
  so we can write $f = f_i + x^{t_i} g$, for some $i$ in $\{1,\dots,\ell\}$ and
  $g$ in $\Kx$.
  Then, $\pol(f) = x^{\val_i}\pol_i(g) = 0 \bmod x^d$, and so
  $g \in \rtset[\pol_i][d-\val_i]$. This implies that $g \in f_{i,j} +
  x^{t_i,j}\Kx$ for some $j$.

  Now consider a power series $g \in \rtset[\pol_i][d-\val_i]$ for some $i$.
  This means that
  $\pol_i(g) = 0 \bmod x^{\max(0, d-\val_i)}$, so that $Q(f_i +
  x^{t_i} g) = x^{s_i}\pol_i(g) = 0 \bmod x^d$, and therefore $f_i +
  x^{t_i} g$ is in $\rtset$.
\end{proof}

We now deduce, by induction on $d$, that any $\pol \in \R$
admits a finite basic root set to precision $d$ for any $d \in
\ZZ_{\geq 0}$. By \cref{lem:modular_roots_valuation} we can reduce to the
case where $\valx(\pol) = 0$ and $\pol\atzero \neq 0$. The
claim is readily seen to be true for $d\le 0$ (take $\{(0,0)\}$) and
$d=1$ (\cref{lem:root_set_1}).  Suppose the claim holds for all $d' <
d$, for some $d \ge 2$; we can then apply this property to
$d-1$, obtaining a basic root set $(f_i, t_i)_{1 \leq i \leq \ell}$ of
$\pol$ to precision $d-1$. We know that, with the notation of \cref{prop:roots_dc}, $s_i \ge d-1$ holds for all $i$, so in particular $s_i
\ge 1$, and thus $d-s_i < d$. Then, applying again the induction
property to each of $(Q_i,d-s_i)_i$, the conclusion of \cref{prop:roots_dc} establishes our claim.

These results can be used to build basic root sets recursively, by either
applying \cref{lem:root_set_1} iteratively or using \cref{prop:roots_dc} in a
divide-and-conquer fashion with \cref{lem:root_set_1} applied at the leaves.
As discussed in \cref{sec:intro}, this recursive approach is similar to the
Newton-Puiseux algorithm. These iterative and divide and conquer solutions to
\cref{pbm:series_root} are known in coding theory as the Roth-Ruckenstein
algorithm \cite{roth_efficient_2000} and the Alekhnovich algorithm
\cite[App.]{alekhnovich_linear_2005}. Below, we describe the latter algorithm
in detail (\cref{alg:dnc_roots}), since our new algorithm runs along the same
lines (\cref{alg:roots}). We will not further discuss the correctness or
complexity of \cref{alg:dnc_roots}, but rather refer to
\cite[App.]{alekhnovich_linear_2005} or
\cite[App.\,A]{nielsen_sub-quadratic_2015}.

\begin{algorithm}[h]
  \caption{\textbf{:} \algoname{DnCSeriesRoots} \cite{alekhnovich_linear_2005}}
  \label{alg:dnc_roots}
  \begin{algorithmic}[1]
    \Require{$\prc \in \ZZp$ and $\pol \in \trcSerPol$ with $\polz \neq 0$.}
    \Ensure{A basic root set of $\pol$ to precision $\prc$.}
    \If{$d = 1$}
      \State $(y_i)_{1 \le i \le \nrts} \ass$ roots of $\polz \in \field[y]$
      \State \Return $(y_i, 1)_{1\le i\le\ell}$
    \Else
      \State\label{alek:5} $(f_i, t_i)_{1\le i\le \nrts} \ass \algoname{DnCSeriesRoots}(\pol \bmod x^{\ceil{\prc/2}},\ceil{\prc/2})$
      \State $(\pol_i)_{1\le i\le \nrts} \ass (\pol(f_i + x^{t_i}y) \bmod x^d)_{1\le i\le \nrts}$
      \State $(\val_i)_{1\le i\le \nrts} \ass (\valx(\pol_i))_{1\le i\le \nrts}$
      \For{$1 \le i \le \nrts$}
        \If{$\val_i \ge \prc$}
          \State $(\rtpol_{i,1},\rtxpt_{i,1}) \ass (0,0)$ and $\nrts_i \ass 1$
        \Else
          \State $(\rtpol_{ij},\rtxpt_{ij})_{1\le j\le \nrts_i} \ass \algoname{DnCSeriesRoots}(x^{-\val_i}\pol_i, \prc-\val_i)$
        \EndIf
      \EndFor
      \State \Return $(\rtpol_i + x^{\rtxpt_i} \rtpol_{i,j}, \rtxpt_i + \rtxpt_{i,j})_{1\le j\le \nrts_i, 1\le i \le \nrts}$.
    \EndIf
  \end{algorithmic}
\end{algorithm}

% Termination and correctness follow again by induction, with the case
% $d=1$ being as in the discussion above. For $d > 1$, we remark that if
% $(f_i, t_i)_{1\le i\le \nrts}$ computed at line~\ref{alek:5} are a
% basic root set of $\pol \bmod x^{\ceil{\prc/2}}$ at precision
% $\ceil{\prc/2}$, they are also a basic root set of $\pol$ at precision
% $\ceil{\prc/2}$. Then, for all $i$, $s_i > \ceil{\prc/2}$, so that
% $d-s_i < d$, and the algorithm finishes. Correctness is a consequence
% of \cref{prop:roots_dc}.

The next step is to prove that there are special, small basic root sets, and that
these also compose in a way similar to \cref{prop:roots_dc}.
In order to formulate this, we first introduce a generalization of root
multiplicity to our setting.
\begin{definition}
  Let $(f, t) \in \K[x] \times \ZZp$ be such that $f$ is nonzero and
  $f=g+ f_{t-1}x^{t-1}$ for some $g\in\K[x]$ with $\deg_x(g) < t-1$. For
  $\pol \in \R\setminus\{0\}$, 
  we consider the polynomial of valuation zero
  \[
    R=Q(g + x^{t-1}y)/x^{\valx(Q(g + x^{t-1}y))} \in \R.
  \]
  Then, the \emph{root
    multiplicity of $(f,t)$ in $Q$} is the root multiplicity of
  $f_{t-1}$ in $R\atzero  \in \K[y]$.
\end{definition}
Note that if $f_{t-1}$ is not a root of $R\atzero$, the root
multiplicity of $(f,t)$ is 0.  Also, if $t=1$, so that $f=f_0$ is in
$\K$, and if $\pol\atzero \ne 0$, the root multiplicity of $(f_0,1)$
is simply the multiplicity of $f_0$ in $\pol\atzero$.

\begin{definition}
  \label{def:reduced_root}
  Let $\pol \in \R$ be such that $\polz \neq 0$ and let $d$ be in
  $\ZZ$. Suppose that $(f_i, t_i)_{1 \leq i \leq \ell}$ is a basic
  root set of $Q$ at precision $d$. Then, we say that $(f_i, t_i)_{1
    \leq i \leq \ell}$ is a \emph{reduced root set}, if the following holds:
  \begin{itemize}[leftmargin=0.5cm]
    \item either $d \leq 0$,
    \item or $d> 0$, and all the $f_i$'s are nonzero, and the following points
      are all satisfied, where for $1 \le i \le \ell$, we write $s_i =
      \valx(Q(f_i + x^{t_i}y))$, $Q_i = Q(f_i + x^{t_i}y)/x^{s_i}$, and we
      write $m_i$ for the root multiplicity of $(f_i, t_i)$ in $Q$:
    \begin{enumerate}
    \item $m_i \geq 1$ for $1\le i\le \ell$,
    \item $\deg({Q_i}\atzero) \leq m_i$ for $1\le i \le\ell$, and
    \item $\sum_{1\le i\le\ell} m_i \;\leq \deg(Q\atzero)$.
    \end{enumerate}
  \end{itemize}
\end{definition}
  
It follows from the restrictions \emph{(1)} and \emph{(3)} that $\ell \leq \deg(Q\atzero)$.
Mimicking the structure of the first half of the section, we now
prove the existence of reduced root sets for $d = 1$ and then give
a composition property. The next lemma
is inspired by~\cite[Lem.\,1.1]{alekhnovich_linear_2005}.

\begin{lemma}
   \label{lem:onestep_roots_partition}
   Let $\pol \in \R$ be such that $\polz \neq 0$. The basic root set
   of $\pol$ to precision $1$ defined in \cref{lem:root_set_1} is reduced.
\end{lemma}
\begin{proof}
  Let $y_1,\ldots,y_\ell$ be the roots of $\pol\atzero$, and, for
  $1\le i\le\ell$, let $\val_i = \valx(Q(y_i+xy))$, $\pol_i =
  Q(y_i+xy)/x^{\val_i}$, and let $m_i$ be the root multiplicity
 of $y_i$ in $\pol\atzero$. 

 The inequalities $m_i \ge 1$, for $1\le i \le\ell$, and $\sum_i m_i \leq
 \deg(Q\atzero)$ are clear. Consider now a fixed index $i$; it remains to prove
 that $\deg({\pol_i}\atzero) \leq m_i$. There are $P \in \K[y]$ and $R \in \R$
 such that $P(y_i) \neq 0$ and $\pol = (y - y_i)^{\mult_i}P(y) + x R$.  Then
  \[
    x^{\val_i} \pol_i = \pol(y_i + xy) = (xy)^{\mult_i}P(y_i + xy) + x R(y_i + xy) \, .
  \]
  The right-hand side reveals the following:
  \begin{itemize}
    \item Any monomial $x^\alpha y^\beta$ in $x^{\val_i}\pol_i$ satisfies
    $\alpha \ge \beta$, and hence $\deg({\pol_i}\atzero) \le \val_i$.
    \item $x^{\val_i} \pol_i$ contains the term $(xy)^{\mult_i} P(y_i)$,
      since this appears in $(xy)^{\mult_i}P(y_i + xy)$ and it cannot
      be cancelled by a term in $xR(y_i + xy)$ since all monomials
      there have greater $x$-degree than $y$-degree.
  \end{itemize}
  These two points imply $\deg({\pol_i}\atzero) \leq s_i \leq m_i$.
\end{proof}

The following theorem is exactly the statement of \cref{prop:roots_dc} except that ``basic'' has been replaced by ``reduced''.

\begin{theorem}
  \label{thm:modular_roots_partition}
  Let $\pol \in \R$ be such that $\polz \neq 0$ and let $ \prc', \prc$
  be in $\ZZ_{\ge 0}$, with $\prc'\le \prc$.  Suppose that $\pol$
  admits a reduced root set $(f_i, t_i)_{1 \leq i \leq \ell}$ to
  precision $\prc'$.  For $i = 1,\ldots,\ell$, suppose furthermore
  that $\pol(f_i + x^{t_i} y)/x^{s_i}$ admits a reduced root set
  $(f_{i,j}, t_{i,j})_{1 \leq j \leq \ell_i}$ to precision $\prc -
  s_i$, where $s_i = \valx(\pol(f_i + x^{t_i} y))$.  Then a reduced root
  set of $\pol$ to precision $\prc$ is given by
  \[
    (f_i + f_{i,j}x^{t_i}, t_i + t_{i,j})_{1 \leq j \leq \ell_i, 1 \leq i \leq \ell} \ .
  \]
\end{theorem}
\begin{proof}
  By \cref{prop:roots_dc} it is clear that the specified set is a basic root set, and we should verify the additional restrictions of \cref{def:reduced_root}.
  Introduce for each $i,j$
  \[
    Q_{i,j} = Q(f_i + f_{i,j}x^{t_i} + x^{t_i + t_{i,j}}y)/x^{s_{i,j}} = Q_i(f_{i,j} + x^{t_{i,j}}y)/x^{s_{i,j}} \ ,
  \]
  where $\pol_i = \pol(f_i + x^{t_i} y)/x^{s_i}$ and $s_{i,j} = \valx(Q_i(f_{i,j} + x^{t_{i,j}}y))$.

  Consider first for some $i$ the case $d - s_i \leq 0$.
  Then $\ell_i = 1$ and $(f_{i,1}, t_{i,1}) = (0,0)$, and so the root multiplicity $m_{i,1}$ of $(f_i + f_{i,1}x^{t_i}, t_{i} + t_{i,1})$ in $Q$ is $m_i$ which is positive by assumption.
  Also $Q_{i,j} = Q_i$ so $\deg({Q_{i,j}}\atzero) = \deg({Q_{i}}\atzero)$ which is at most $m_i = m_{i,1}$ by assumption.
  Finally, $\sum_j m_{i,j} = m_{i,1} = m_i$.
  We will collect the latter fact momentarily to prove the third item of the reduced root definition.

  Consider next an $i$ where $d - s_i > 0$.
  In this case $t_{i,j} > 0$ for all $1 \leq j \leq \ell_i$, and the root multiplicity of $(f_i + f_{i,j}x^{t_i}, t_i + t_{i,j})$
  in $Q$ equals the root multiplicity $m_{i,j}$ of $(f_{i,j}, t_{i,j})$ in $Q_i$ which is positive by assumption.
  The assumptions also ensure that $\deg({Q_{i,j}}\atzero) \leq m_{i,j}$, and $\sum_j m_{i,j} \leq \deg({Q_i}\atzero) \leq m_i$.

  Thus, the two first restrictions on being a reduced root set is satisfied for each element.
  All that remains is the third restriction: but using our previous observations, we have $\sum_i \sum_j m_{i,j} \leq \sum_i m_i$ and this is at most $\deg(Q\atzero)$ by assumption.
\end{proof}

To solve \cref{pbm:series_root} we will compute a reduced root set
using \cref{lem:onestep_roots_partition} and
\cref{thm:modular_roots_partition}.  Note that it follows that a
reduced root set is essentially unique: apart from possible redundant
elements among the $f_i$, non-uniqueness would only be due to
unnecessarily expanding a coefficient in a root $(f,t)$, that is, replace
that root by the $|\K|$ roots $(f + ax^t, t+1)_{a \in \K}$.  Of course
this could only be an issue if $\K$ is finite and if
$\deg(\pol\atzero)$ is very large.  Our algorithm as well as previous
ones are computing the ``minimal'' set of reduced roots. According to
\cref{thm:modular_roots_partition}, the total number of field elements
required to represent this minimal set cannot exceed $ \prc
\deg(Q\atzero) \le \prc\deg(Q)$.

%% \todo{not entirely obvious, but by induction, all $t_i \le d$}
%% \jsrn{Should we address bounded precision in $x$ in recursive calls somewhere? This is implicit in the correctness of our algorithm.}

\section{Affine factors of the shifts}
\label{sec:affine_factors}

The appendix~A of~\cite{nielsen_sub-quadratic_2015} gives a careful
complexity analysis of \cref{alg:dnc_roots}, and proves that it runs
in time $\Osoft(d n^2 + d n \univRoots)$, where $n=\deg(\pol)$.
The main reason why the cost is
quadratic in $\deg(\pol)$ is that all the shifted polynomials $\pol_i
= x^{-\val_i}\pol(\rtpol_i + x^{\rtxpt_i}y)$ can have large degree,
namely up to $\deg(\pol)$. Thus, merely representing the $\pol_i$'s
may use a number of field elements quadratic in $\deg(\pol)$.

Nonetheless, we are actually not interested in these shifts
themselves, but only in their reduced root sets. The number of these
roots is well controlled: the shifts have altogether a reduced root
set of at most $\deg(\polz)$ elements. Indeed, by definition, we know
that $\deg({\pol_i}\atzero)$ is at most the multiplicity $\mult_i$ of
the root $(\rtpol_i,\rtxpt_i)$, and the sum of these multiplicities is
at most $\deg(\polz)$.

The difficulty we face now is that we want to efficiently compute
reduced root sets of the shifts without fully computing these
shifts. To achieve this, we compute for each shift $\pol_i$ a factor
of it which has the same roots and whose degree is
$\deg({\pol_i}\atzero) \le \mult_i$, {\em without entirely computing $\pol_i$
  itself}. We design a fast algorithm for computing these factors, by
using ideas from \cite[Algo.\,Q]{Musser75}, in which we also
incorporate fast modular reduction techniques so as to carefully
control the quantity of information we process concerning the shifts.

The next result formalizes the factorization we will rely on. It is a
direct consequence of the Weierstrass preparation theorem for
multivariate power series \cite[VII.\S1.~Cor.\,1 of Thm.\,5]{ZarSam60}.
\begin{theorem}
  \label{thm:affine_factor}
  Let $\pol \in \Kx[y]$ be such that $\polz \neq 0$. Then, there exist
  unique $\aff, \coaff \in \Kx[y]$ such that $\pol = \aff \coaff$,
  $\aff$ is monic and $\coaff\atzero \in \field\setminus\{0\}$.
\end{theorem}
%% ---> best thing is to write a clear proof (not necessarily in the paper, but at least for us)

In the case at hand, one may as well derive existence and uniqueness
of $A$ and $B$ (together with a slow algorithm to compute them) by
writing their unknown coefficients as $A=a_0(y)+ x a_1(y) + \cdots$
$B=b_0+xb_1(y)+\cdots$, with $b_0$ in $\K\setminus\{0\}$ and all $a_i$'s ($i
\ge 1$) of degree less than that of $a_0$. Extracting coefficients of
$x^0,x^1,\dots$, we deduce that the relation $Q=AB$ defines the $a_i$'s
and $b_i$'s uniquely.

In what follows, $\aff$ is called the \emph{affine factor} of $\pol$.
Remark that if we start from $\pol$ in $\trcSerPol$, we can still
define its affine factor as a polynomial in $\trcSerPol$, by reducing
modulo $x^d$ the affine factor of an arbitrary lift of $\pol$ to
$\Kx[y]$ (the construction above shows that the result is independent 
of the lift).

Our algorithm will compute the affine factors $(\aff_i)_{1\le i\le
  \nrts}$ of the shifts $(\pol_i)_{1\le i\le \nrts}$ at some
prescribed precision $d$ in $x$, having as input $\pol$ and the
shifting elements $(\rtpol_i + x^{\rtxpt_i}y)_{1\le i\le \nrts}$. A
factorization $\pol_i = \aff_i \coaff_i$ can be computed modulo any
power $x^d$ of $x$ from the knowledge of $\pol_i$ by means of Hensel
lifting \cite[Algo.\,Q]{Musser75}, doubling the precision at each
iteration. However, the above-mentioned degree bounds indicate that
neither the shifts $(\pol_i)_i$ nor the cofactors $(\coaff_i)_i$ may
be computed modulo $x^d$ in time quasi-linear in $\deg(\pol)$ and
$\prc$: the key of our algorithm is to show how to compute the
affine factors $\aff_i$ at precision $d$ directly from $Q$ within the
prescribed time bounds.  (Hensel lifting factorization techniques were
also used in~\cite{berthomieu_polynomial_2013}, but in a context
without the degree constraints that prevent us from computing the
shifts $\pol_i$). Hereafter, $A \quo B$ and $A \rem B$ denote the
quotient and the remainder in the division of the polynomial $A$ by
the monic polynomial $B$.

The input of the algorithm is the polynomial $Q$ known modulo $x^d$,
as output, we compute the affine factors $A_i$ of the shifts at
respective precisions $d-s_i$, together with the valuation $s_i$; if
$s_i \ge d$, we detect it and return $(0,d)$. 
The initialization consists in computing the affine factors of the
$x$-constant polynomials $({\pol_i}\atzero)_{1\le i\le \nrts}$. If
these polynomials are known, this is straightforward: the affine
factor of ${\pol_i}\atzero$ is itself divided by its leading
coefficient, which is a nonzero constant from $\field$. It turns out
that computing these polynomials is not an issue; remark that the sum of
their degrees is at most $\mult_1 + \cdots + \mult_\nrts \le
\deg(\pol)$. Explicitly, we first compute the remainders $(\pol(\rtpol_i +
x^{\rtxpt_i} y) \rem y^{\mult_i+1})_i$ via fast modular reduction
techniques; then, we can both retrieve the valuations $(\val_i)_i =
(\valx(\pol(\rtpol_i + x^{\rtxpt_i} y)))_i$ (or, more precisely,
$\val^*_i=\min(\val_i, d)$), and, when $\val_i < d$, the $x$-constant
terms of $\pol_i=x^{-\val_i} \pol(\rtpol_i + x^{\rtxpt_i} y)$ to carry
out the initialization step (\cref{line:init} to \cref{line:endinit}
in \cref{alg:affine_factors}).

Before continuing to describe the algorithm, we detail one of its main
building blocks (\cref{alg:shifted_remaindering}): the fast
computation of simultaneous shifted remainders via multiple modular
reduction. 

\begin{algorithm}[h]
  \caption{\textbf{:} \algoname{ShiftedRem}}
  \label{alg:shifted_remaindering}
  \begin{algorithmic}[1]
    \Require
      a commutative ring $\ring$, a polynomial $\pol \in \ring[y]$, and triples
      $(\aff_i,\rt_i,r_i)_{1\le i\le \nrts} \in \ring[y] \times \ring \times
      \ring$, with the $\aff_i$'s monic.
    \Ensure
      the remainders $\pol(\rt_i + r_i y) \rem \aff_i$ for $1\le i \le \nrts$.
      \State $(\bar\aff_i)_{1\le i\le \nrts} \ass (\sum_{0\le j\le\delta_i} r_i^{\delta_i-j} a_{i,j} y^j)_{1\le i\le\nrts}$ \\
      \hfill where $(\aff_i)_{1\le i\le \nrts} = (\sum_{0\le j\le\delta_i} a_{i,j} y^j)_{1\le i\le\nrts}$ with $a_{i,\delta_i}=1$
    \State $(\hat\aff_i)_{1\le i\le \nrts} \ass (\bar\aff_i(y-\rt_i))_{1\le i\le \nrts}$
    \State $(\hat\rmd_i)_{1\le i\le \nrts} \ass$ $(\pol \rem \hat\aff_i)_{1 \le i \le \nrts}$
    \State $(\rmd_i)_{1\le i\le\nrts} \ass (\hat \rmd_i(\rt_i+r_i y))_{1 \le i \le \nrts}$
    \State \Return $(\rmd_i)_{1\le i\le\nrts}$
  \end{algorithmic}
\end{algorithm}

\begin{proposition}
  \label{prop:shifted_remaindering}
  \cref{alg:shifted_remaindering} is correct and uses
  \[
    \Osoft{(\deg(\pol) + \deg(\aff_1 \cdots \aff_\nrts))}
  \]
  operations in $\ring$.
\end{proposition}
\begin{proof}
  Let $i \in \{1,\ldots,\nrts\}$. Since $\hat\aff_i$ is monic, the remainder
  $\hat\rmd_i = \pol \rem \hat\aff_i$ is well-defined, and $\pol = P_i \hat\aff_i
  + \hat\rmd_i$ with $\deg(\hat\rmd_i) < \deg(\pol)$ and $P_i \in \ring[y]$. Then,
  we have
  \begin{align*}
    \pol(\rt_i + r_i y) & = P_i(\rt_i + r_i y) \hat\aff_i(\rt_i + r_i y) + \hat\rmd_i(\rt_i + r_i y) \\
      & = P_i(\rt_i + r_i y) \bar\aff_i(r_i y) + \rmd_i(y) \\
      & = P_i(\rt_i + r_i y) r_i^\delta \aff_i(y) + \rmd_i(y),
  \end{align*}
  which ensures $\rmd_i = \pol(\rt_i + r_i y) \rem \aff_i(y)$, hence the
  correctness.

  Concerning the cost bound, the polynomial $\bar\aff_i$ is computed using at
  most $2\delta_i$ multiplications in $\ring$, where $\delta_i = \deg(\aff_i)$,
  and then $\hat\aff_i$ is computed by fast shifting using $\Osoft(\delta_i)$
  operations in $\ring$ \cite[Thm.\,9.15]{von_zur_gathen_modern_2013}. The
  conclusion follows, since fast remaindering can be used to compute all
  remainders $(\hat\rmd_1,\ldots,\hat\rmd_\nrts)$ simultaneously in
  $\Osoft{(\deg(\pol) + \delta_1 + \cdots + \delta_\nrts)}$ operations in
  $\ring$. Indeed, we start by computing the subproduct tree in
  $\Osoft(\delta_1+\cdots+\delta_\nrts)$ operations
  \cite[Lem.\,10.4]{von_zur_gathen_modern_2013}, which gives us in particular
  the product $\hat\aff_1 \cdots \hat\aff_\nrts$. Then, we compute the
  remainder $\hat\rmd = \pol \rem \hat\aff_1 \cdots \hat\aff_\nrts$, which can
  be done in $\Osoft{(\deg(\pol) + \delta_1 + \cdots + \delta_\nrts)}$
  operations in $\ring$ using fast division
  \cite[Thm.\,9.6]{von_zur_gathen_modern_2013}. Finally, the sought
  $\hat\rmd_i = \hat\rmd \bmod \hat\aff_i$ are computed by going down the
  subproduct tree, which costs $\Osoft{(\delta_1 + \cdots + \delta_\nrts)}$
  operations in $\ring$ \cite[Cor.\,10.17]{von_zur_gathen_modern_2013}.
\end{proof}

\begin{algorithm}[t]
  \caption{\textbf{:} \algoname{AffineFacOfShifts}}
  \label{alg:affine_factors}
  \begin{algorithmic}[1]
    \Require
      a precision $\prc\in\ZZp$, a polynomial $\pol \in \trcSerPol$ such that
      $\polz \neq 0$, and triples $(\rtpol_i,\rtxpt_i,\mult_i)_{1\le i\le \nrts}
      \subset \trcSer \times \ZZp \times \ZZp$.
    \Ensure
      $(\aff_i,\val_i)_{1 \le i \le \nrts}$ with $(\aff_i,\val_i)= (0,\prc)$ if
      $\pol(\rtpol_i + x^{\rtxpt_i} y) = 0$ in $\trcSerPol[\prc]$, and otherwise $\val_i =
      \valx(\pol(\rtpol_i + x^{\rtxpt_i} y)) < \prc$ and $\aff_i \in
      \trcSerPol[\prc-\val_i]$ is the affine factor of
      $Q_i=x^{-\val_i}\pol(\rtpol_i + x^{\rtxpt_i} y)$ at precision $d-s_i$.
    \Assume 
      $\mult_i$ is such that
      $\aff_i = 0$ or $\deg(\aff_i) \le \mult_i$, for $1 \le i \le \nrts$.

      \State \label{line:init} $\alive \ass (1,\ldots,\nrts)$ \hfill /* \emph{list of not yet computed factors} */
      \State \label{line:init_shiftedrem} $(\rmd_i)_{1 \le i\le \nrts} \ass \algoname{ShiftedRem}(\trcSer,\pol,(y^{\mult_i+1},\rtpol_i,x^{\rtxpt_i})_{1 \le i\le \nrts})$
      \State /* \emph{Process trivial affine factors} */
      \For{$1 \le i \le \nrts$ such that $\rmd_i = 0$}
        \State $(\aff_i,\val_i) \ass (0,\prc)$, and
               remove $i$ from $\alive$
      \EndFor
      \State /* \emph{Set valuations and compute affine factors $\bmod \:x$} */
      \For{$i \in \alive$}
        \State $\val_i \ass \valx(\rmd_i)$
        \State $\bar\rmd_i \in \K[y] \,\ass (x^{-\val_i}\rmd_i)\atzero$
        \State $\coaffi_i \in \field\setminus\{0\} \,\ass $ inverse of the leading coefficient of $\bar\rmd_i$
        \State $\aff_i \in \trcSerPol[1] \,\ass \coaffi_i \bar\rmd_i$
      \EndFor \label{line:endinit}
      \State /* \emph{Each iteration doubles the precision} */
      \For{$1 \le k \le \lceil\log_2(\prc)\rceil$}
        \For{$i \in \alive$ such that $\prc-\val_i \le 2^{k-1}$}
          \State remove $i$ from $\alive$
        \EndFor
        \State $K \ass 2^{k-1} ; \;\; (\delta_i)_{i\in\alive} \ass (\min(K,\prc-\val_i - K))_{i\in\alive}$
        \State \rule[0.03cm]{0pt}{\baselineskip} /* \emph{Lift the affine factors $(\aff_i)_i$ to precisions $\delta_i+K$} */
        \State \label{line:rem_a} $(\rmd_i)_{i\in\alive} \ass \algoname{ShiftedRem}(\trcSer,\pol,(\bar\aff_i,\rtpol_i,x^{\rtxpt_i})_{i\in\alive})$, \\
          \hspace{1.1cm} \rule[-0.15cm]{0pt}{\baselineskip} where $\bar\aff_i$ is $\aff_i$ lifted into $\trcSerPol$
        \State \label{line:lift_a_start} $(\aff_{i\top} \in \trcSerPol[\delta_i])_{i\in\alive} \ass ((x^{-\val_i-K} \rmd_i \coaffi_i) \rem \aff_i)_{i\in\alive}$, \\
          \hspace{1.1cm} \rule[-0.15cm]{0pt}{\baselineskip} with $x^{-\val_i-K} \rmd_i$, $\coaffi_i$, and $\aff_i$ truncated at precision $\delta_i$
        \State \label{line:lift_a_end} $(\aff_i \in \trcSerPol[\delta_i+K])_{i\in\alive} \;\ass (\aff_i + x^{K} \aff_{i\top})_{i\in\alive}$
        \State \rule[0.05cm]{0pt}{\baselineskip} /* \emph{Find the cofactor inverses $(\coaffi_i)_i$ at precisions $\delta_i+K$} */
        \State \label{line:rem_aa} $(S_i)_{i\in\alive}  \ass \algoname{ShiftedRem}(\trcSer,\pol,(\bar\aff_i^2,\rtpol_i,x^{\rtxpt_i})_{i\in\alive})$, \\
          \hspace{1.1cm} \rule[-0.15cm]{0pt}{\baselineskip} where $\bar\aff_i$ is $\aff_i$ lifted in $\trcSerPol$
        \State \label{line:lift_c} $(\coaffi_i \in \trcSerPol[\delta_i+K])_{i\in\alive} \ass (((x^{-\val_i} S_i) \quo \aff_i)^{-1} \rem \aff_i)_{i\in\alive}$, \\
          \hspace{1.1cm} \rule[-0.1cm]{0pt}{\baselineskip} with $x^{-\val_i} S_i$ and $\aff_i$ truncated at precision $\delta_i+K$
      \EndFor
    \State \Return $(\aff_i,\val_i)_{1\le i\le\nrts}$
  \end{algorithmic}
\end{algorithm}

Now, let us describe how we implement the Hensel lifting strategy to manage to
compute the sought affine factors without fully computing the shifts. In
addition to the affine factors, we will make use of partial information on the
inverse of the cofactor: we compute this inverse modulo the affine factor.
Let $1 \le i \le \nrts$ and assume that we have computed, at precision $K$,
\begin{itemize}
  \item the affine factor $\aff_i \in \trcSerPol[K]$ of $\pol_i \bmod x^{K}$,
  \item $\coaffi_i = \coaff_i^{-1} \rem \aff_i \in \trcSerPol[K]$, where
    $\coaff_i\in \trcSerPol[K]$ denotes the cofactor such that $\aff_i \coaff_i
    = \pol_i \bmod x^{K}$.
\end{itemize}
Note that $\coaff_i$ is invertible as a polynomial of $\trcSerPol[K]$ since by
definition ${\coaff_i}\atzero \in \field\setminus\{0\}$. Thus, our requirement
is that the inverse of $\coaff_i$ coincides with $\coaffi_i$ when working modulo
$\aff_i$.

Now, we want to find similar polynomials when we increase the
precision to $2K$. The main point concerning efficiency is that we
will be able to do this by only considering computations modulo the
affine factors $\aff_i$ and their squares; remember that we control
the sum of their degrees. In the algorithm, we increase for each $i$
the precision from $K$ to $K+\delta_i$, which is taken as the minimum of $2K$
and $\prc-\val_i$: in the latter case, this is the last iteration
which affects $\aff_i$, since it will be known at the wanted precision
$\prc-\val_i$.

First, we use fast remaindering to get $R_i= \pol(\rtpol_i + x^{\rtxpt_i}
y) \rem \aff_i$ at precision $d$ in $x$, simultaneously for all $i$
(see \cref{line:rem_a}); this gives us $ \pol_i \rem
\aff_i=x^{-\val_i}R_i\rem \aff_i $ at
precision $d-\val_i$, and thus $K+\delta_i$. Since $\aff_i$ is the
affine factor of $\pol_i$ at precision $K$, $\pol_i \rem \aff_i$ is
divisible by~$x^K$.

We then look for $\aff_{i\top} \in \trcSerPol[\delta_i]$ such that
$\hat\aff_i = \aff_i + x^K \aff_{i\top}$ is the affine factor of
$\pol_i$ at precision $K+\delta_i$; to ensure that $\hat\aff_i$ is still
monic, we require that $\deg(\aff_{i\top}) < \deg(\aff_i)$. Thus, we
can determine $\aff_{i\top}$ by working modulo $\aff_i$: having
\[
  (\aff_i + x^K \aff_{i\top}) (\coaff_i + x^K \coaff_{i\top}) = \pol_i ,
\]
at precision $K+\delta_i$, for some $\coaff_{i\top}\in \trcSerPol[\delta_i]$, implies that the
identity
\[
  \aff_{i\top} \coaff_i = x^{-K} \pol_i
\]
holds modulo $\aff_i$ and at precision $\delta_i$. Multiplying by $\coaffi_i =
\coaff_i^{-1}$ on both sides yields
\[
  \aff_{i\top} = (x^{-K} \pol_i \coaffi_i) \rem \aff_i = (x^{-K-\val_i} \rmd_i \coaffi_i) \rem \aff_i \ .
\]
Therefore,
\cref{line:lift_a_start} and \cref{line:lift_a_end} correctly lift the affine
factor of $\pol_i$ from precision $K$ to precision $K+\delta_i$.

\medskip

From now on, we work at precision $K+\delta_i$, and, as in the
pseudo-code, we denote by $\aff_i$ the affine factor obtained through
the lifting step above (that is, $\aff_i \ass \hat\aff_i$). Besides,
let $\coaffi_i$ now denote the cofactor inverse at precision
$K+\delta_i$: $\coaffi_i = \coaff_i^{-1} \rem \aff_i$, where $\coaff_i
\in \trcSerPol[K+\delta_i]$ is the cofactor such that $\pol_i = \aff_i
\coaff_i$. Our goal is to compute $C_i$, without computing $B_i$ but only
$\coaff_i \rem \aff_i$.

We remark that the remainder $S_i = \pol(\rtpol_i+x^{\rtxpt_i}y) \rem
\aff_i^2$ (as in \cref{line:rem_aa}) is such that $x^{-\val_i} S_i =
\pol_i \rem \aff_i^2 = \aff_i (\coaff_i \rem \aff_i)$; $x^{-\val_i}
S_i$ is known at precision $d-s_i \ge K+\delta_i$. Thus,
\[
  (x^{-\val_i} S_i) \quo \aff_i = \coaff_i \rem \aff_i \ ,
\]
and therefore
$\coaffi_i$ can be obtained as
\[
  \coaffi_i = \coaff_i^{-1} \rem \aff_i = ((x^{-\val_i} S_i) \quo \aff_i)^{-1} \rem \aff_i.
\]
This shows that \cref{line:lift_c} correctly computes $\coaffi_i$ at precision
$K+\delta_i$.

\begin{proposition}
  \label{prop:compl_AffineFactorsOfShifts}
  \cref{alg:affine_factors} is correct and uses
  \[
    \Osoft\big(\prc(\deg(\pol) + \mult_1 + \cdots + \mult_\nrts)\big)
  \]
  operations in $\field$.
\end{proposition}
\begin{proof}
  The correctness follows from the above discussion. Concerning the cost bound,
  we will use the following degree properties. Since $\aff_i$ is monic, we have
  the degree bound $\deg(\aff_i) = \deg({\aff_i}\atzero) \le \mult_i$ for all
  $i$ and at any iteration of the loop; and since $\coaffi_i$ is always
  computed modulo $\aff_i$, we also have $\deg(\coaffi_i) < \mult_i$.
  
  The cost of the initialization (\cref{line:init} to \cref{line:endinit}) is
  dominated by the computation of shifted remainders at
  \cref{line:init_shiftedrem}, which costs $\Osoft(\prc(\deg(\pol) + \mult_1 +
  \cdots + \mult_\nrts))$ operations in $\field$ according to
  \cref{prop:shifted_remaindering}. The same cost bound holds for each call to
  \algoname{ShiftedRem} at \cref{line:rem_a} or \cref{line:rem_aa}, since we
  have $\deg(\aff_i) \le \mult_i$ and $\deg(\aff_i^2) \le 2\mult_i$.
  
  At both \cref{line:lift_a_start} and \cref{line:lift_c}, the degrees of
  $\rmd_i$, $\coaffi_i$, and $\aff_i$ are at most $\mult_i$; besides, we have
  $\delta_i \le \prc$ and $\delta_i+K \le 2\prc$. Thus, the quotient and
  remainder computations use $\Osoft(\prc (\mult_1+\cdots+\mult_\nrts))$
  operations in $\field$ according to
  \cite[Thm.\,9.6]{von_zur_gathen_modern_2013}.

  Finally, at \cref{line:lift_c} we are performing the inversion of the
  polynomial $((x^{-\val_i} \rmd_i) \quo \aff_i)$ modulo $A_i$; it is invertible in
  $\trcSerPol[\delta_i+K]/(A_i)$ since its $x$-constant coefficient is a nonzero
  field element. As a consequence, this this inversion can be done in
  $\Osoft((\delta_i+K) \deg(\aff_i))$ field operations using Newton iteration
  \cite[Thm.\,9.4]{von_zur_gathen_modern_2013}, and altogether
  \cref{line:lift_c} costs $\Osoft(\prc (\mult_1+\cdots+\mult_\nrts))$
  operations in $\field$.

  Summing these cost bounds over the $\lceil\log_2(\prc)\rceil$
  iterations yields the announced total cost bound.
\end{proof}

\section{Fast series roots algorithm}
\label{sec:alg:roots}

In this section, we describe our fast algorithm for solving
\cref{pbm:series_root}. As explained above, it follows the divide and conquer
strategy of \cref{alg:dnc_roots}, with the main modification being that we
incorporate the fast computation of the affine factors of the shifts
(\cref{alg:affine_factors}). This leads to a better efficiency by yielding more
control on the degrees of the polynomials that are passed as arguments to the
recursive calls. Besides, we also propagate in recursive calls the information
of the multiplicities of the roots, which is then used as an input of
\cref{alg:affine_factors} to specify the list of degree upper bounds for the
affine factors.

We start with a lemma which states that taking affine factors preserves
reduced root sets.

\begin{lemma}\label{lem:root_set_aff_fact}
  Let $Q$ be in $\Kx[y]$, with $Q\atzero \ne 0$, and let $A \in
  \Kx[y]$ be its affine factor. Then, any reduced root set of $A$ at
  precision $d$ is a reduced root set of $Q$ at precision $d$.
\end{lemma}
\begin{proof}
  The claim follows from the factorization $Q=AB$, with $B \atzero \in
  \K\setminus\{0\}$. Indeed, as a result, $B(P)$ is a unit in $\R$ for
  any $P$ in $\R$, hence $\rtset = \rtset[\aff]$ for any $d$;
  similarly, for any $(f,t)$, $Q(f+x^t y)$ and $A(f+x^t y)$ have the
  same valuation, say $s$, and $Q(f+x^t y)/x^s$ and $A(f+x^t y)/x^s$
  differ by a constant factor. 
  %TODO what??
  In particular, if $\{(f_i,t_i)\}_i$ is
  a basic root set of $A$, it is a basic root set of $Q$, and the
  multiplicities of $(f_i,t_i)$ in $A$ and $Q$ are the same.  This
  implies that if $\{(f_i,t_i)\}_i$ is in fact a reduced root set of
  $A$, it remains so for $Q$.
\end{proof}

We continue with a procedure that operates on polynomials in $\Kx[y]$,
without applying any truncation with respect to $x$: as such, this is
not an algorithm over $\K$, as it defines objects that are power
series in $x$, but it is straightforward to prove that it outputs a
reduced root set. Remark that this procedure uses affine factors at
``full precision'', that is, in $\Kx[y]$,
so \cref{alg:affine_factors} is not used yet.

\begin{algorithm}[h]
  \caption{\textbf{:} \algoname{SeriesRoots\infty}}
  \label{alg:roots}
  \begin{algorithmic}[1]
    \Require{$\prc \in \ZZp$ and $\pol \in \Kx[y]$ such that $\polz \neq 0$.}
    \Ensure{List of triples $(\rtpol_i, \rtxpt_i, \mult_i)_{1 \le i \le \nrts}
      \subset \K[x] \times \NN \times \ZZp$ formed by a reduced root set
      of $\pol$ to precision $\prc$ with multiplicities.}
    \If{$d = 1$}
      \State $(y_i,\mult_i)_{1\le i\le\nrts} \ass $ roots with multiplicity of $Q\atzero \in \field[y]$
      \State \Return $(y_i, 1, \mult_i)_{1 \le i \le \ell}$
    \Else
      \State $(f_i, t_i, \mult_i)_{1 \le i \le \nrts} \ass \algoname{SeriesRoots\infty}(Q, \ceil{d/2})$
      \State $(s_i)_{1 \le i \le \nrts} \ass (\valx(Q(f_i+x^{t_i} y))_{1 \le i \le \nrts}$
      \State $(\aff_i)_{1 \le i \le \nrts} \ass ({\rm AffineFactor}(Q(f_i+x^{t_i} y)/x^{s_i}))_{1 \le i \le \nrts}$
 %% \algoname{AffineFacOfShifts}(Q, d, (\rtpol_i,\rtxpt_i,\mult_i)_{1\le i\le\nrts})$
      \For{$1 \le i \le \nrts$}
        \If{$\val_i \ge \prc$}
          \State $(\rtpol_{i,1},\rtxpt_{i,1},\mult_{i,1}) \ass (0,0,\mult_i)$ and $\nrts_i \ass 1$
        \Else
          \State $(\rtpol_{i,j},\rtxpt_{i,j},\mult_{i,j})_{1\le j\le \nrts_i} \ass \algoname{SeriesRoots\infty}(\aff_i, \prc-\val_i)$
          \label{line:roots_recursive2}
        \EndIf
      \EndFor
      \State \Return $(\rtpol_i + x^{\rtxpt_i} \rtpol_{i,j}, \rtxpt_i + \rtxpt_{i,j},\mult_{i,j})_{1\le j\le \nrts_i, 1\le i \le \nrts}$.
    \EndIf
  \end{algorithmic}
\end{algorithm}

\begin{proposition}
  \label{prop:alg:roots}
  \cref{alg:roots} is correct.
\end{proposition}
\begin{proof}
  We prove this by induction on $d \ge 1$. By
  \cref{lem:onestep_roots_partition}, the algorithm is correct for the
  induction base case $d = 1$. Take $d > 1$, and assume that the
  algorithm is correct for all $d' < d$. Then, we obtain a reduced root set
  $(f_i, t_i)$ from the first recursive call, so in particular the
  valuations $s_i$ are at least equal to $d' \ge 1$. This shows that
  $d-s_i < d$, so the second recursive call is made at a lower
  precision, and the procedure terminates.

  By induction, in all cases, $(\rtpol_{i,j},\rtxpt_{i,j})_{1\le j\le
    \nrts_i}$ is a reduced root set of $Q_i$ to precision $d-s_i$:
  this is obvious when $s_i \ge d$, and follows from
  \cref{lem:root_set_aff_fact} when $s_i <
  d$. \cref{thm:modular_roots_partition} implies that $(\rtpol_i +
  x^{\rtxpt_i} \rtpol_{i,j}, \rtxpt_i + \rtxpt_{i,j})_{1\le j\le
    \nrts_i, 1\le i \le \nrts}$ is a reduced root set of $\pol$ to
  precision $\prc$. We verify that the integers $m_{i,j}$ are the associated
  multiplicities as we did in the proof of that theorem.
\end{proof}

Next, we describe a similar algorithm, where we maintain the input
polynomial with degree less than $d$ in $x$ (when it is the case, we
say that it is {\em reduced modulo $x^d$}). To differentiate this
version from the previous one and facilitate proving the correctness, we
add a superscript ${}^*$ to the objects handled here when they differ
from their counterpart in \cref{alg:roots}. Remark that we do not
claim that the output forms a reduced root set of $Q^*$, merely a
basic root set; we also do not claim that the $m_i$'s in the output are
the corresponding multiplicities.

\begin{algorithm}[h]
  \caption{\textbf{:} \algoname{SeriesRootsTrc}}
  \label{alg:rootsTrc}
  \begin{algorithmic}[1]
    \Require{$\prc \in \ZZp$ and $\pol^* \in \Kx[y]$ reduced modulo $x^d$ such that $\pol^*\atzero \neq 0$.}
    \Ensure{List of triples $(\rtpol_i, \rtxpt_i, \mult_i)_{1 \le i \le \nrts}
      \subset \K[x] \times \NN \times \ZZp$ formed by a basic root set
      of $\pol^*$ to precision $\prc$.}
    \If{$d = 1$}
      \State $(y_i,\mult_i)_{1\le i\le\nrts} \ass $ roots with multiplicity of $Q^*\atzero \in \field[y]$
      \State \Return $(y_i, 1, \mult_i)_{1 \le i \le \ell}$
    \Else
      \State $(f_i, t_i, \mult_i)_{1 \le i \le \nrts} \ass \algoname{SeriesRootsTrc}(Q^* \rem x^{\ceil{d/2}}, \ceil{d/2})$
      %% \State $(s_i)_{1 \le i \le \nrts} \ass (\valx(Q(f_i+x^{t_i} y))_{1 \le i \le \nrts}$
      \State $(\aff^*_i, \val^*_i)_{1 \le i \le \nrts} \ass \algoname{AffineFacOfShifts}(Q^*, d, (\rtpol_i,\rtxpt_i,\mult_i)_{1\le i\le\nrts})$
      \For{$1 \le i \le \nrts$}
        \If{$\val^*_i = \prc$}
          \State $(\rtpol_{i,1},\rtxpt_{i,1},\mult_{i,1}) \ass (0,0,\mult_i)$ and $\nrts_i \ass 1$
        \Else
          \State $(\rtpol_{i,j},\rtxpt_{i,j},\mult_{i,j})_{1\le j\le \nrts_i} \ass \algoname{SeriesRootsTrc}(\aff^*_i, \prc-\val^*_i)$
          \label{line:roots_recursive2Trc}
        \EndIf
      \EndFor
      \State \Return $(\rtpol_i + x^{\rtxpt_i} \rtpol_{i,j}, \rtxpt_i + \rtxpt_{i,j},\mult_{i,j})_{1\le j\le \nrts_i, 1\le i \le \nrts}$.
    \EndIf
  \end{algorithmic}
\end{algorithm}

\begin{proposition}
  \label{prop:alg:rootstrunc}
  \cref{alg:rootsTrc} is correct.
\end{proposition}
\begin{proof}
  We claim that for $d > 0$ and any $Q$ and $Q^*$ in $\Kx[y]$ such that
  $Q^* = Q \rem x^d$, the outputs of
  \algoname{SeriesRoots\infty}$(Q,d)$ and
  \algoname{SeriesRootsTrc}$(Q^*,d)$ are the same. Before proving this
  claim, remark that it implies the correctness of \cref{alg:rootsTrc}: we
  know that this output is a reduced, and thus basic, root set of $Q$
  to precision $d$.
  Since $Q$ and $Q^*$ are equal modulo $x^d$, one verifies easily that this
  output is thus a basic root set of $Q^*$ to precision $d$ as well.

  We prove the claim by induction on $d$. If $d=1$, the result is clear,
  as we compute the same thing on both sides.

  For $d > 1$, since $Q^* \rem x^{\lceil d/2\rceil}=Q \rem x^{\lceil
    d/2\rceil}$, the induction assumption shows that $(f_i, t_i,
  \mult_i)_{1 \le i \le \nrts} $ as computed in either
  \algoname{SeriesRoots\infty} or \algoname{SeriesRootsTrc} are the
  same. 

  The affine factors of the shifts of $Q$ and $Q^*$ differ, but they
  coincide at the precision we need.  Indeed, the equality $Q=Q^*
  \bmod x^d$ implies that for all $i$, $Q(f_i + x^{t_i}y)= Q^*(f_i +
  x^{t_i}y) \bmod x^d$. In particular, if $s_i < d$, these two
  polynomials have the same valuation $s_i$, and $Q(f_i +
  x^{t_i}y)/x^{s_i}=Q^*(f_i + x^{t_i}y)/x^{s_i} \bmod x^{d-s_i}$,
  which implies that their affine factors are the same modulo
  $x^{d-s_i}$. If $s_i \ge d$, then $Q^*(f_i + x^{t_i}y)$ vanishes
  modulo $x^d$.

  Remark that the assumption of \cref{alg:affine_factors} is
  satisfied: for all $i$, $m_i$ is the multiplicity of $(f_i,t_i)$ in
  $Q$; the definition of a reduced root set then implies that
  $\deg({Q_i}\atzero) \le m_i$, so that the same degree bounds holds
  for the affine factors of $Q^*(f_i + x^{t_i}y)/x^{s_i}$. As a
  result, for $i$ such that $s_i \ge d$, \cref{alg:affine_factors}
  returns $(0,s_i^*)=(0,d)$, whereas if $s_i < d$, it returns
  $(A_i^*,s_i)$, where $A_i^*$ is the truncation modulo $x^{d-s_i}$ of
  the affine factor $A_i$ of $Q_i$. In the first case, the polynomials
  $(\rtpol_{i,1},\rtxpt_{i,1},\mult_{i,1})$ are the same in both
  algorithms; in the second case, this is also true, by induction
  assumption. Our claim follows.
\end{proof}
    
\begin{proof}[Proof of \cref{thm:series_root}]
  To conclude the proof of \cref{thm:series_root}, it remains to
  estimate the cost of \cref{alg:rootsTrc}. Let $T(n, d)$ denote the
  cost of \cref{alg:rootsTrc} on input $d$ and $\pol$ of degree $n = \deg(Q)$.
  If $d = 1$, then $T(n, 1) = \univRoots(n)$.
  Otherwise, the cost is given by the following recursion:
  \[
    T(n, d) = T(n, d/2) + S(n, d, (n_1,\ldots,n_\ell)) + \sum_{i=1}^{\ell} T(n_i, d-s_i) \ ,
  \]
  where $S(n,d,(n_1,\ldots,n_\ell))$ is the cost of
  \algoname{AffineFactorsOfShifts} and $n_i = \deg(A^*_i)$. The
  degrees of the polynomials $A^*_i$, in \cref{alg:rootsTrc}, and $A_i$,
  in \cref{alg:roots}, are the same, except for those cases where $s_i
  \ge d$ and $A^*_i$ is actually zero.  By definition of a reduced root
  set, we have $$\sum_i \deg(A_i) \leq \deg(Q\atzero) \leq n,$$ which
  thus implies $\sum_i n_i \le n$, and $S(n, d, (n_1,\ldots,n_\ell)) \in
  \Osoft(dn)$. Note also that $s_i \geq d/2$ by the correctness of
  $\algoname{SeriesRootsTrc}$. Since $T(n,d)$ is at least linear in $n$,
  we then get $\sum_i T(n_i, d-s_i) \leq T(n, d/2)$.
  This gives the upper bound
  \[
    T(n,d) \le 2 T(n,d/2) + \Osoft(nd),
  \]
  from which we deduce that $T(n,d) = \Osoft(nd) + O(d \univRoots(n))$.
\end{proof}

Finally, we point out an optimization, which is not necessary to
establish our main result, but useful in practice: once the affine
factor of a shift has degree $1$, there is no need to continue the
recursion (the affine factor being monic, we can just read off its
root from its constant coefficient).  This is the analogue of the
situation described in the introduction, when we know enough terms of
the solution to make it possible to apply Newton iteration without
further branching.

\begin{acks}
  The research leading to these results has received funding from the People
  Programme (Marie Curie Actions) of the European Union's Seventh Framework
  Programme (FP7/2007-2013) under REA grant agreement no. 609405
  (COFUNDPostdocDTU).
\end{acks}

\bibliographystyle{ACM-Reference-Format}
%%% -*-BibTeX-*-
%%% Do NOT edit. File created by BibTeX with style
%%% ACM-Reference-Format-Journals [18-Jan-2012].

\end{document}